\documentclass[a4paper]{llncs}

\usepackage{graphicx}
\usepackage{amssymb}
\usepackage{amsmath}
\usepackage{enumerate}
\usepackage{scalefnt}
\usepackage{setspace} 

\usepackage{url}
\newcommand{\keywords}[1]{\par\addvspace\baselineskip
\noindent\keywordname\enspace\ignorespaces#1}

\newcommand{\argmax}{\operatornamewithlimits{argmax}}

\newtheorem{thm}{Theorem}
\newtheorem{lem}{Lemma}
\newtheorem{clm}{Claim}
\newtheorem{coro}{Corollary}

\begin{document}

\setstretch{0.96}

\mainmatter

\title{Minimax Regret 1-Median Problem \\ in Dynamic Path Networks}

\author{
Yuya Higashikawa \inst{1,6}
\and Siu-Wing Cheng \inst{2}
\and Tsunehiko Kameda \inst{3}
\and Naoki Katoh \inst{4,6}\thanks{Supported by JSPS Grant-in-Aid for Scientific Research(A)(25240004)}
\and Shun Saburi \inst{5}
}

\institute{
{Department of Information and System Engineering, Chuo University, Japan, 
higashikawa.874@g.chuo-u.ac.jp}
\and
{Department of Computer Science and Engineering, The Hong Kong University of Science and Technology, Hong Kong, 
scheng@cse.ust.hk}
\and
{School of Computing Science, Simon Fraser University, Canada,
tiko@sfu.ca}
\and
{Department of Informatics, Kwansei Gakuin University, Japan,
naoki.katoh@gmail.com}
\and
{Department of Architecture and Architectural Engineering, Kyoto University, Japan,
as-saburi@archi.kyoto-u.ac.jp}
\and
{CREST, Japan Science and Technology Agency (JST), Japan}
}

\maketitle

\begin{abstract}
This paper considers the minimax regret 1-median problem in dynamic path networks.
In our model, we are given a dynamic path network consisting of an undirected path with positive edge lengths, uniform positive edge capacity,
and nonnegative vertex supplies.
Here, each vertex supply is unknown but only an interval of supply is known.
A particular assignment of supply to each vertex is called a {\it scenario}. 
%
Given a scenario $s$ and 
a sink location
$x$ in a dynamic path network,
let us consider the evacuation time to $x$ of a unit supply
given on a vertex by $s$. 
The cost of $x$ under $s$ is defined as the sum of evacuation times to $x$ for all supplies given by $s$, 
and the {\it median} under $s$ is defined as a sink location which minimizes this cost.
The regret for $x$ under $s$ is defined as the cost of $x$ under $s$ minus the cost of the median under $s$. 
Then, the problem is to find a sink location such that the maximum regret for all possible scenarios is minimized.
%
%
We propose an $O(n^3)$ time algorithm for the minimax regret 1-median problem in dynamic path networks with uniform capacity, 
where $n$ is the number of vertices in the network. 
\end{abstract}

\section{Introduction}
The Tohoku-Pacific Ocean Earthquake happened in Japan on March 11, 2011, 
and many people failed to evacuate and lost their lives due to severe attack by tsunamis. 
From the viewpoint of disaster prevention from city planning and evacuation planning,  
it has now become extremely important to establish effective evacuation planning systems against large scale disasters in Japan. 
In particular, arrangements of tsunami evacuation buildings in large Japanese cities near the coast has become an urgent issue. 
To determine appropriate tsunami evacuation buildings, we need to consider where evacuation buildings are located 
and how to partition a large area into small regions so that one evacuation building is designated in each region. 
This produces several theoretical issues to be considered. 
Among them, this paper focuses on the location problem of the evacuation building assuming that we fix the region such that all evacuees in the region are planned to evacuate to this building. 
In this paper, we consider the simplest case for which the region consists of a single road. \\
\indent
In order to represent the evacuation, we consider the {\it dynamic} setting in graph networks, which was first introduced by Ford et al.~\cite{ff58}.
In a graph network under the dynamic setting, each vertex is given supply and each edge is given length and capacity which limits the rate of the flow into the edge per unit time.
We call such networks under the dynamic setting {\it dynamic networks}.
Unlike in static networks, the time required to move supply from one vertex to a sink can be increased due to congestion caused by the capacity constraints, 
which require supplies to wait at vertices until supplies preceding them have left.
In this paper, we consider the flow on dynamic networks as continuous, that is, each input value is given as a real number, and supply, flow and time are defined continuously.
Then each supply can be regarded as fluid, and edge capacity is defined as the maximum amount of supply which can enter an edge per unit time.
The {\it 1-sink location problem in dynamic networks} is defined as the problem which requires to find the optimal location of a sink in a given dynamic network 
so that all supplies are sent to the sink as quickly as possible. \\
\indent
In order to evaluate an evacuation, we can naturally consider two types of criteria: {\it completion time criterion} and {\it total time criterion}.
In this paper we adopt the latter one (for the former one, refer to \cite{h14,hgk14_2,hgk14_4,mumf06}).
We here define a {\it unit} as an infinitesimally small portion of supply.
Given a sink location $x$ in a dynamic network, let us consider an evacuation to $x$ starting at time $0$ 
and define the {\it evacuation time} of a unit to $x$ as the time at which the unit reaches $x$ in the evacuation.
The total time for the evacuation to $x$ is defined as the sum of evacuation times over all infinitesimal units to $x$.
Then, the minimum total time for all possible evacuations to $x$ could be the criterion for the optimality of sink location, which we adopt.
Given a dynamic network, we define the {\it 1-median problem} as the problem which requires to find a sink location minimizing the minimum total time, 
and the optimal solution is called the {\it median}.\\
\indent
Although the above criterion is reasonable for the sink location, it may not be practical since the number of evacuees in an area may vary depending on the time
(e.g., in an office area in a big city, there are many people during the daytime on weekdays while there are much less people on weekends or during the night time). 
So, in order to take into account the uncertainty of population distribution,
we consider the {\it maximum regret} for a sink location as another evaluation criterion assuming that for each vertex, we only know an interval of vertex supply. 
Then, the {\it minimax regret 1-median problem in dynamic path networks} is formulated as follows.
A particular assignment of supply to each vertex is called a {\it scenario}.
Here, for a sink location $x$ and a scenario $s$, we denote the minimum total time by $\Phi^s(x)$.
Also let $m^s$ denote the median under $s$.
The problem can be understood as a 2-person Stackelberg game as follows.
The first player picks a sink location $x$ and the second player chooses a scenario $s$ that maximizes the {\it regret} defined as $\Phi^s(x)-\Phi^s(m^s)$.
The objective of the first player is to choose $x$ that minimizes the maximum regret. \\
\indent
Related to the minimax regret facility location in graph networks, especially for trees, some efficient algorithms have been presented by~\cite{ab00,bk12,bks13,bgk08,cl98,c08}.
For dynamic networks, Cheng et al.~\cite{chknsx13} first studied the
{\it minimax regret 1-center problem} in path networks,
which requires to find a sink location in a path that minimizes the maximum regret where the completion time criterion is adopted instead of the total time one.
They presented an $O(n \log^2 n)$ time algorithm.
Higashikawa et al.~\cite{hacknsx14} improved the time bound by \cite{chknsx13} to $O(n \log n)$, and also Wang~\cite{w13} independently achieved the same time bound of $O(n \log n)$ with better space complexity.
Very recently, Bhattacharya et al.~\cite{bk14} have improved the time bound to $O(n)$.
The above problem was extended to the multiple sink location version by Arumugam et al.~\cite{aags14}.
For the minimax regret $k$-center problem in dynamic path networks with uniform capacity, they presented an $O(kn^3 \log n)$ time algorithm, and this time bound was improved to $O(kn^3)$ recently~\cite{h14}.
On the other hand, for dynamic tree networks, only the minimax regret 1-center problem was solved in $O(n^2 \log^2 n)$ time~\cite{hgk14,hgk14_3}.

This paper first considers the minimax regret median problem in dynamic networks while all the above works for dynamic networks treated center problems.
In this paper, we address the minimax regret 1-median problem in dynamic path networks with uniform capacity and present an $O(n^3)$ time algorithm. \\


\section{Preliminaries}
\label{sec:pre}

\subsection{Dynamic path networks under uncertain supplies}
\label{subsec:dpn}
Let $P =(V, E)$ be an undirected path with ordered vertices  $V = \{ v_1, v_2, \ldots, v_n \}$ and  edges $E = \{ e_1, e_2, \ldots, e_{n-1} \}$ where
 $e_i=(v_i,v_{i+1})$ for $i \in \{1, \ldots, n-1\}$.
Let $\mathcal{N} = (P, l, w, c, \tau)$ be a dynamic network 
with the underlying path graph $P$; 
$l$ is a function that associates each edge $e_i$ with positive length $l_i$, 
$w$ is a function that associates each vertex $v_i$ with positive weight $w_i$, amount of supply at $v_i$;
$c$ is the capacity, a positive constant representing the amount of supply which can enter an edge per unit time;
$\tau$ is a positive constant representing the time required for a flow to travel a unit distance.
In our model, instead of the weight function $w$ on vertices, we are given the weight interval function $W$
that associates each vertex $v_i \in V$ with an interval of supply $W_i = [w^-_i, w^+_i]$ with $0 < w^-_i \leq w^+_i$.
We call such a network $\mathcal{N} = (P, l, W, c, \tau)$ with path structures {\it a dynamic path network under uncertain supplies}.

In the following, we write $p \in P$ to indicate that a point is a vertex of $P$ or lies on one of the edges of $P$.  
For any point $p \in P$,
we abuse this notation by also letting $p$ denote the distance from $v_1$ to $p$. 
Informally we can regard $P$ as being embedded on a real line with $v_1 = 0$.
For two points $p, q \in P$ satisfying $p < q$, 
let $[p, q]$ (resp. $[p, q)$, $(p, q]$ and $(p, q)$) denote an interval in $P$ consisting of all points $x \in P$ 
such that $p \le x \le q$ (resp. $p \le x < q$, $p < x \le q$ and $p < x < q$).

\subsection{Scenarios}
\label{subsec:scn}
Let ${\cal S}$ denote the Cartesian product of all $W_i$ for $i \in \{1, \ldots, n\}$:
\begin{eqnarray}
{\cal S} = \prod_{i=1}^n W_i.
\label{eq:sc}
\end{eqnarray}
An element of ${\cal S}$, i.e., a particular assignment of weight to each vertex, is called a {\it scenario}.
Given a scenario $s \in {\cal S}$, we denote by $w^s_i$ the weight of a vertex $v_i$ under $s$.

\subsection{Total evacuation time}
\label{subsec:to}
In our model, the supply is defined continuously.
We define a unit as an infinitesimally small portion of supply.
Given a sink location $x \in P$ and a scenario $s \in {\cal S}$, 
without loss of generality, an evacuation to $x$ under $s$ is assumed to satisfy the following assumptions.
When a unit arrives at a vertex $v$ on its way to $x$, it has to wait for the departure if there are already some units waiting for leaving $v$. 
All units waiting at $v$ for leaving $v$ are processed in the first-come first-served manner. 

For a given $x \in P$ and $s \in {\cal S}$, let us consider an evacuation to $x$ under $s$ starting at time $0$
and define the evacuation time of a unit to $x$ under $s$ as the time at which the unit reaches $x$.
Let $\Phi^s(x)$ denote the sum of evacuation times over all infinitesimal units to $x$ under $s$.
Also let $\Phi^s_L(x)$ (resp. $\Phi^s_R(x)$) denote the sum of evacuation times to $x$ under $s$ for all units on $[v_1, x)$ (resp. $(x, v_n]$).
Then, $\Phi^s(x)$ is obviously the sum of $\Phi^s_L(x)$ and $\Phi^s_R(x)$, i.e., 
\begin{eqnarray}
\Phi^s(x) = \Phi^s_L(x) + \Phi^s_R(x). \label{eq:ms1}
\end{eqnarray}
Without loss of generality, we assume $\Phi^s_L(v_1) = 0$ and $\Phi^s_R(v_n) = 0$.

We will show the formula of $\Phi^s(x)$ that has been proved in \cite{hgk14_2,hgk14_4}.
Suppose that $x$ is located in an open interval $(v_h, v_{h+1})$ with $1 \le h \le n-1$, i.e., $x \in e_h$.
We here show only the formula of $\Phi^s_L(x)$ (the case of $\Phi^s_R(x)$ is symmetric).
First, let us define the vertex indices $\rho_1, \ldots, \rho_k$ inductively as
\begin{eqnarray}
	\rho_i = \argmax \left\{ \tau(v_h - v_j) + \frac{\sum_{l=\rho_{i-1}+1}^j w^s_l}{c} \ \bigg| \ j \in \{\rho_{i-1}+1, \ldots, h\} \right\},
\label{ms2.5}
\end{eqnarray}
where $\rho_0 = 0$.
Obviously $\rho_k = h$ holds.
We then call a set of all units on $[v_{\rho_{i-1}+1},v_{\rho_i}]$ {\it i-th cluster},
and a vertex $v_{\rho_i}$ the {\it head} of $i$-th cluster (see Fig.~\ref{fig:cl}).
\begin{figure}[b]
\centering
\includegraphics[width=90mm,clip]{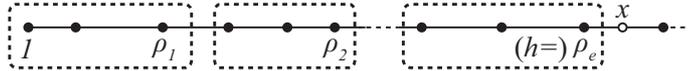} 
\caption{Illustration of left clusters for $x$ where $i$-th cluster is headed by a vertex $v_{\rho_i}$}
\label{fig:cl}
\end{figure}

\noindent
Also, for $i \in \{1, \ldots, k\}$, we define $\sigma_i$ as $\sigma_i = \sum_{l =\rho_{i-1}+1}^{\rho_i} w^s_l$,
which is called the weight of $i$-th cluster.
The interpretation can be derived from \cite{hgk14_2,hgk14_4} as follows.
The first unit on each $v_{\rho_i}$ does not encounter any congestion on its way to $x$. 
Here, although $\rho_i$ may not be uniquely determined by (\ref{ms2.5}),
we choose the maximum index as $\rho_i$.
By this assumption and (\ref{ms2.5}), the following inequality holds for $i \in \{2, \ldots, k\}$:
\begin{eqnarray}
\tau(v_{\rho_i}-v_{\rho_{i-1}}) > \frac{\sigma_i}{c}.
\label{eq:ms2.6}
\end{eqnarray}
In other words, even if we transform the input so that all units on $[v_{\rho_{i-1}+1},v_{\rho_i})$ are moved to $v_{\rho_i}$ for $i \in \{1, \ldots, k\}$,
the sum of evacuation times over $[v_1,x)$ never changes.
In the following, we call such a transformation the {\it left-clustering} for $x$ and
clusters obtained by the left-clustering for $x$ {\it left clusters} for $x$
(the right-clustering and right clusters are symmetric).
Then as in \cite{hgk14_2,hgk14_4}, $\Phi^s_L(x)$ is written as 
\begin{eqnarray}
\Phi^s_L(x)  	&=& \sum_{i=1}^k \left( \sigma_i\tau(x-v_{\rho_i}) + \frac{{\sigma_i}^2}{2c} \right). 
\label{eq:ms3}
\end{eqnarray}

\subsection{Minimax regret formulation}
\label{subsec:mi}
For a scenario $s \in {\cal S}$, let $m^s$ be a point in $P$ that minimizes $\Phi^s(x)$ over $x \in P$, called the {\it median} under $s$. 
We now define the {\it regret} for $x$ under $s$ as
\begin{eqnarray}
R^s(x) &=&  \Phi^s(x) - \Phi^s(m^s).
\label{eq:rg}
\end{eqnarray}
Moreover, we also define the {\it maximum regret} for $x$ as
\begin{eqnarray}
R_{\rm max}(x) &=& \max \{ R^s(x) \mid s \in \mathcal{S} \}.
\label{eq:rmax}
\end{eqnarray}
If $\hat{s} = \argmax \{ R^s(x) \mid s \in \mathcal{S} \}$, 
we call $\hat{s}$ the {\it worst case scenario} for $x$.
The goal is to find a point $x^* \in P$, called the {\it minimax regret median}, that minimizes $R_{\rm max}(x)$ over $x \in P$,
i.e., the objective is to
\begin{eqnarray}
{\rm minimize} \ \{ R_{\rm max}(x) \mid x \in P\}.
\label{eq:objective}
\end{eqnarray}

\subsection{Known properties for the fixed scenario case}
\label{subsec:kp}
We here show some properties on the 1-median problem in a dynamic path network $\mathcal{N} = (P = (V, E), l, w^s, c, \tau)$ when a scenario $s \in {\cal S}$ is given, which were basically presented in \cite{hgk14_2,hgk14_4}.
We first introduce the following two lemmas.

\begin{lem}\hspace{-1mm}{\rm \cite{hgk14_2,hgk14_4}}
For a scenario $s \in {\cal S}$, $m^s$ is at a vertex in $V$.
\label{lem:ms}
\end{lem}


\begin{lem}\hspace{-1mm}{\rm \cite{hgk14_2,hgk14_4}}
For a scenario $s \in {\cal S}$, all $\Phi^s(v_i)$ over $i \in \{1, \ldots, n\}$ can be computed in $O(n)$ time in total.
\label{lem:pht}
\end{lem}

\noindent
We then can see a corollary of these lemmas.

\begin{coro}\hspace{-1mm}{\rm \cite{hgk14_2,hgk14_4}}
For a scenario $s \in {\cal S}$, $m^s$ and $\Phi^s(m^s)$ can be computed in $O(n)$ time.
\label{coro:mst}
\end{coro}


Now let us look at the formula of (\ref{eq:ms3}).
Even if $x$ is moving on an edge $e_i$ (not including endpoints $v_i$ and $v_{i+1}$), the formation of left clusters for $x$ does not change over $x \in e_i$.
Therefore, $\Phi^s_L(x)$ is a linear function of $x \in e_i$, and symmetrically, $\Phi^s_R(x)$ is also a linear function.
For $i \in \{1, \ldots, n-1\}$, letting $a^s_i$ and $b^s_i$ be the values such that for $x \in e_i$,
\begin{eqnarray}
\Phi^s(x) = a^s_ix + b^s_i,
\label{eq:aibi}
\end{eqnarray}
we can derive the following lemma from \cite{hgk14_2,hgk14_4}.

\begin{lem}
For a scenario $s \in {\cal S}$, all $a^s_i$ and $b^s_i$ over $i \in \{1, \ldots, n-1\}$ can be computed in $O(n)$ time in total.
\label{lem:abt}
\end{lem}

\section{Properties of worst case scenarios}
\label{sec:pwcs}

In this section, we show the important properties which worst case scenarios have.
In our problem, a main difficulty lies in evaluating $R^s(x)$ over $s \in {\cal S}$ to compute $R_{\rm max}(x)$ even for a fixed $x$ since 
the size of ${\cal S}$ is infinite.
We thus aim to find a scenario set with a finite size (in particular, a polynomial size)
which includes a worst case scenario for any $x \in P$.
In order to do this, we introduce a new concept, the {\it gap} between two points $x, y \in P$ under a scenario $s \in {\cal S}$, defined by
\begin{eqnarray}
\Gamma^s(x, y) = \Phi^s(x) - \Phi^s(y).
\label{eq:gap}
\end{eqnarray}
By Lemma \ref{lem:ms} and the definition of (\ref{eq:rg}), we have
\begin{eqnarray}
R^s(x) = \max \{ \Gamma^s(x, y) \mid y \in V \}, 
\label{eq:rg2}
\end{eqnarray}
and by (\ref{eq:rmax}) and (\ref{eq:rg2}),
\begin{eqnarray}
R_{\rm max}(x)	&=& \max \{ \max \{\Gamma^s(x, y) \mid y \in V \} \mid s \in {\cal S} \} \notag \\
			&=& \max \{ \max \{\Gamma^s(x, y) \mid s \in {\cal S} \} \mid y \in V \}.
\label{eq:rmax2}
\end{eqnarray}
From (\ref{eq:rmax2}), if we can compute $\max \{\Gamma^s(x, y) \mid s \in {\cal S} \}$ for a fixed pair $\langle x, y \rangle \in P \times V$,
$R_{\rm max}(x)$ can also be computed by repeating the same maximization over $y \in V$.
We call a scenario that maximizes $\Gamma^s(x, y)$ for a fixed $\langle x, y \rangle$ a worst case scenario for $\langle x, y \rangle$.
In the following, we show a scenario set of size $O(n)$ that includes a worst case scenario for a fixed $\langle x, y \rangle$,
which implies a scenario set of size $O(n^2)$ that includes a worst case scenario for a fixed $x$.

\subsection{Bipartite scenario}
\label{subsec:bs}
We first introduce the concept of the {\it bipartite scenario}, which was originally introduced as the {\it dominant scenario} in \cite{chknsx13,hacknsx14}.
Let us consider a scenario $s \in {\cal S}$.
A scenario $s$ is said to be {\it left-bipartite} (resp. {\it right-bipartite})
if $w^s_j = w^+_j$ (resp. $w^-_j$) over $j \in \{1, \ldots, i\}$ and 
$w^s_j = w^-_j$ (resp. $w^+_j$) over $j \in \{i+1, \ldots, n\}$ for some $i \in \{1, \ldots, n-1\}$.
Obviously the number of such scenarios is $O(n)$.
The authors of \cite{chknsx13,hacknsx14} treated the minimax regret 1-center problem in dynamic path networks,
which requires to find a sink location in a path that minimizes the maximum regret similarly defined as (\ref{eq:rmax})
where the completion time criterion is adopted instead of the total time one.
They proved that for any point in an input path, at least one worst case scenario is left-bipartite or right-bipartite.

\subsection{Pseudo-bipartite scenario}
\label{subsec:pbs}
We here introduce the concept of the {\it pseudo-bipartite scenario}.
A scenario $s$ is said to be {\it left-pseudo-bipartite} (resp. {\it right-pseudo-bipartite})
if $w^s_j = w^+_j$ (resp. $w^-_j$) over $j \in \{1, \ldots, i-1\}$ and 
$w^s_j = w^-_j$ (resp. $w^+_j$) over $j \in \{i+1, \ldots, n\}$ for some $i \in \{2, \ldots, n-1\}$.
In this definition, we do not care about the weight of a vertex $v_i$, called the {\it intermediate vertex}.
Given a pseudo-bipartite scenario with the intermediate vertex $v_i$, we call intervals $[v_1, v_i)$ and $(v_i, v_n]$ the {\it left part} and the {\it right part}, respectively.
Let ${\cal S}_L$ (resp. ${\cal S}_R$) denote a set of all left-pseudo-bipartite scenarios (resp. right-pseudo-bipartite scenarios).
We then prove the following lemma (the proof is given in Appendix A).
\begin{lem}
Given a pair $\langle x, y \rangle \in P \times V$ satisfying $y < x$ (resp. $x < y$), there exists a worst case scenario for $\langle x, y \rangle$ belonging to ${\cal S}_L$ (resp. ${\cal S}_R$) such that $y$ (resp. $x$) is in the left part and $x$ (resp. $y$) is in the right part.
\label{lem:wcsxy}
\end{lem}

\subsection{Critical pseudo-bipartite scenario}
\label{subsec:cpbs}
By Lemma \ref{lem:wcsxy}, we studied the property of a worst case scenario for a fixed $\langle x, y \rangle \in P \times V$,
however 
the sizes of ${\cal S}_L$ and ${\cal S}_R$
are still infinite since the weight of the intermediate vertex in a pseudo-bipartite scenario is not fixed.
In the rest of this section, we focus on the weight of the intermediate vertex in a pseudo-bipartite scenario
which is worst for $\langle x, y \rangle$.

Given a pair $\langle x, y \rangle \in P \times V$ satisfying $y < x$,
let us consider a scenario $s \in {\cal S}_L$ 
such that the intermediate vertex is $v_i$ and $y<v_i<x$.
Suppose that the weight of $v_i$ is set as the minimum, i.e., $w^s_i=w^-_i$.
Performing the right-clustering for $y$ under $s$ (mentioned in Section \ref{subsec:to}), 
we will get right clusters for $y$ such that 
for $l \in \{1, \ldots, k\}$,
the head of $l$-th cluster is $\rho_l$
and the weight of $l$-th cluster is $\sigma_l$.
Then, suppose that the intermediate vertex $v_i$ belongs to $j$-th cluster.

Now let us increase the weight of $v_i$, little by little, without changing the weight of any other vertex.
Let 
$s(w)$ be a scenario in ${\cal S}_L$ such that the intermediate vertex is $v_i$ whose weight is $w \in [w^-_i, w^+_i]$.
Suppose that when the weight of $v_i$ reaches some value $\omega$, the following equality holds:
\begin{eqnarray}
\tau(v_{\rho_{j-1}}-v_{\rho_j}) = \frac{\sigma_j+(\omega-w^-_i)}{c}.
\label{eq:cpbs1} 
\end{eqnarray}
Note that $\sigma_j+(\omega-w^-_i)$ corresponds to the weight of $j$-th cluster under $s(\omega)$.
At that moment, referring to (\ref{eq:ms2.6}), $(j-1)$-th cluster is merged to $j$-th cluster.
We then call $s(\omega)$ a {\it critical left-pseudo-bipartite scenario} for $y$.
Also, 
$s(w^-_i)$ and $s(w^+_i)$ are assumed to be critical left-pseudo-bipartite scenarios for $y$
even if any merge does not occur at those moments.
{\it Critical right-pseudo-bipartite scenarios} for $y$ are symmetrically defined.
Let ${\cal S}_y$ denote a set of all critical left-pseudo-bipartite scenarios and critical right-pseudo-bipartite scenarios for $y$,
and ${\cal S}^* = \bigcup_{y \in V}{\cal S}_y$.
We will show two lemmas
(the proof of Lemma \ref{lem:wcsxy2} is given in Appendix B).

\begin{lem}
Given a pair $\langle x, y \rangle \in P \times V$, there exists a worst case scenario for $\langle x, y \rangle$ belonging to ${\cal S}_y$.
\label{lem:wcsxy2}
\end{lem}
\begin{lem}
For a vertex $y \in V$, 
the size of ${\cal S}_y$ is $O(n)$,
and all scenarios in ${\cal S}_y$ can be computed in $O(n)$ time.
\label{lem:csy}
\end{lem}
\begin{proof}
We first prove that 
the number of critical left-pseudo-bipartite scenarios for $y$ is $O(n)$
(the case of critical right-pseudo-bipartite scenarios is symmetric).
Suppose that $y = v_j$.
For $i \in \{j+1, \ldots, n\}$ and $w \in [w^-_i, w^+_i]$, let 
$s(i, w)$ be a scenario in ${\cal S}_L$
such that the intermediate vertex is $v_i$ whose weight is $w$.
Here, let us define the order between two scenarios $s(i, w)$ and $s(i', w')$:
$s(i, w) \prec s(i', w')$ holds if and only if 
(a) $i < i'$ or 
(b) $i = i'$ and $w < w'$. 
For $i \in \{j+1, \ldots, n\}$, we also define $p_i$ and $q_i$ as follows.
Let $p_i$ be the number of critical left-pseudo-bipartite scenarios for $y$ such that the intermediate vertex is $v_i$
(including $s(i, w^-_i)$ and $s(i, w^+_i)$).
Let $q_i$ be, under a scenario $s(i, w^+_i)$, the number of right clusters for $y$ that follow a cluster including $v_i$.

Let us consider computing all critical left-pseudo-bipartite scenarios for $y$ in ascending order,
and suppose that the weight of $v_i$ now increases from $w^-_i$ to $w^+_i$.
While it increases,
since $p_i-2$ critical left-pseudo-bipartite scenarios for $y$ occur (except $s(i, w^-_i)$ and $s(i, w^+_i)$),
and at each such scenario, one or more clusters are merged into one of $v_i$,
at least $p_i-2$ clusters are merged into one of $v_i$ in total.
We thus have $q_i \le q_{i-1} - (p_i-2)$ for $i \in \{j+1, \ldots, n\}$, i.e.,
\begin{eqnarray}
p_i \le q_{i-1} - q_i + 2.
\label{eq:lem1.21} 
\end{eqnarray}
Note that the total number of critical left-pseudo-bipartite scenarios for $y$ is exactly $1+\sum_{i=j+1}^n (p_i-1)$.
By (\ref{eq:lem1.21}), we have 
\begin{eqnarray}
\sum_{i=j+1}^n (p_i-1) \le \sum_{i=j+1}^n (q_{i-1} - q_i + 1) = q_j - q_n + (n-j),
\label{eq:lem1.22}
\end{eqnarray}
which is $O(n)$ since $q_j \le n-j$ and $q_n = 0$.

In the rest of the proof, we show that all critical left-pseudo-bipartite scenarios for $y = v_j$ can be computed in $O(n)$ time.
Recall that all critical left-pseudo-bipartite scenarios for $y$ are computed in ascending order.
The algorithm first gets $s(j+1, w^-_{j+1})$,
and performs the right clustering for $y$ under $s(j+1, w^-_{j+1})$.
As claimed in \cite{hgk14_2,hgk14_4}, it is easy to see that the right clustering for a fixed $y$ can be done in $O(n)$ time.

Suppose that for particular $i \in \{j+1, \ldots, n\}$ and $\omega \in [w^-_i, w^+_i]$, $s(i, \omega)$ is critical for $y$,
and the algorithm has already obtained $s(i, \omega)$ and the right clusters for $y$.
We then show how to compute the subsequent critical left-pseudo-bipartite scenario.
Let $c_y$ be a right cluster for $y$ including $v_i$ and $c'_y$ be a right cluster for $y$ immediately following $c_y$.
Also, let $\rho_y$ (resp. $\rho'_y$) be the index of a vertex that corresponds to the head of $c_y$ (resp. $c'_y$),
and $\sigma_y$ (resp. $\sigma'_y$) be the weight of $c_y$ (resp. $c'_y$).

There are two cases: 
[Case 1] $\omega < w^+_i$; 
[Case 2] $\omega = w^+_i$.
For Case 2,
we notice that $s(i+1, w^-_{i+1})$ is equivalent to $s(i, w^+_i)$.
Therefore, this case immediately results in Case 1 by letting $i$ be $i+1$ and $\omega$ be $w^-_{i+1}$
(although a right cluster for $y$ including $v_{i+1}$ may be $c'_y$, not $c_y$).
We thus consider only Case 1 in the following.

The algorithm will compute the subsequent critical left-pseudo-bipartite scenario $s(i, \omega')$
where $\omega'$ satisfies $\omega < \omega' \le w^+_i$.
In order to compute $\omega'$, the algorithm test if there exists $w \in (\omega, w^+_i]$ such that
\begin{eqnarray}
\tau(v_{\rho'_y}-v_{\rho_y})=\frac{\sigma_y + (w - \omega)}{c},
\label{eq:lem1.23}
\end{eqnarray}
which is similar to (\ref{eq:cpbs1}).
If yes, for such $w$, the algorithm returns $\omega'=w$ and updates the right clusters for $y$ by merging $c'_y$ into $c_y$.
Otherwise, $\omega'=w^+_i$ is just returned. 
Such testing and updating are done in $O(1)$ time.
Since the number of critical left-pseudo-bipartite scenarios for $y$ is $O(n)$ and each of those is computed in $O(1)$ time,
we completes the proof.
\qed
\end{proof}
By (\ref{eq:rmax2}), we have a corollary of Lemma \ref{lem:wcsxy2}.
\begin{coro}
Given a point $x \in P$, there exists a worst case scenario for $x$ belonging to ${\cal S}^*$.
\label{coro:wcsx}
\end{coro}
Also, a corollary of Lemma \ref{lem:csy} immediately follows.
\begin{coro}
The size of ${\cal S}^*$ is $O(n^2)$,
and all scenarios in ${\cal S}^*$ can be computed in $O(n^2)$ time.
\label{coro:cusy}
\end{coro}

\section{Algorithm}
\label{sec:al}
In this section, we show an algorithm that computes the minimax regret median, which minimizes $R_{\rm max}(x)$ over $x \in P$.
The algorithm basically consists of two phases: \\
{\bf [Phase 1]} Compute $R_{\rm max}(v_i)$ over $i \in \{1, \ldots, n\}$, and \\
{\bf [Phase 2]} Compute $\min \{R_{\rm max}(x) \mid x \in e_i\}$ over $i \in \{1, \ldots, n-1\}$. \\
After these, the algorithm evaluates all the $2n-1$ values obtained and finds the minimax regret median in $O(n)$ time.

By Corollary \ref{coro:wcsx}, we only have to consider scenarios in ${\cal S}^*$ to compute $R_{\rm max}(x)$ for any $x \in P$.
Therefore, the algorithm computes all scenarios in ${\cal S}^*$ in advance, which can be done in $O(n^2)$ time by Corollary \ref{coro:cusy}.
Subsequently, it computes all the values $\Phi^s(m^s)$ over $s \in {\cal S}^*$ for Phase 1 and Phase 2.
By Corollaries \ref{coro:mst} and \ref{coro:cusy}, this can be done in $O(n^3)$ time in total.  

First let us see details in Phase 1. 
For a fixed scenario $s \in {\cal S}^*$,
since all $\Phi^s(v_i)$ over $i \in \{1, \ldots, n\}$ can be computed in $O(n)$ time by Lemma \ref{lem:pht}
and $\Phi^s(m^s)$ has already been computed before Phase 1,
all $R^s(v_i)$ over $i \in \{1, \ldots, n\}$ can also be computed in $O(n)$ time (refer to (\ref{eq:rg})).
After the algorithm obtains $R^s(v_1), \ldots, R^s(v_n)$ over $s \in {\cal S}^*$ in $O(n^3)$ time,
for each $i \in \{1, \ldots, n\}$,
$R^s(v_i)$ over $s \in {\cal S}^*$ are evaluated to obtain $R_{\rm max}(v_i)$.
Thus, it is easy to see that Phase 1 can be done in $O(n^3)$ time in total.

We next focus on Phase 2. 
As mentioned at the end of Section \ref{subsec:kp}, 
for a fixed scenario $s \in {\cal S}^*$, 
$\Phi^s(x)$ is a linear function of $x \in e_i$ for each $i \in \{1, \ldots, n-1\}$ (not including $v_i$ and $v_{i+1}$).
Therefore, $R^s(x)$ is also linear for $x \in e_i$ for each $i$.
Referring to (\ref{eq:aibi}), a function $R^s(x)$ on an edge $e_i$ is written as
\begin{eqnarray}
R^s(x) = a^s_ix + b^s_i - \Phi^s(m^s).
\label{eq:aibi2}
\end{eqnarray}
Recall that $\Phi^s(m^s)$ has already been computed.
Then, by Lemma \ref{lem:abt}, $R^s(x)$ on $e_i$ over $i \in \{1, \ldots, n-1\}$ can be computed in $O(n)$ time.
After the algorithm does the same computation over $s \in {\cal S}^*$ in $O(n^3)$ time,
on each edge $e_i$,
we have $O(n^2)$ linear functions $R^s(x)$ over $s \in {\cal S}^*$.
By the definition of (\ref{eq:rmax}), $\min\{R_{\rm max}(x) \mid x \in e_i\}$ can be obtained by solving a linear programming problem in two dimensions with $O(n^2)$ constraints,
i.e.,
\begin{eqnarray}
\mbox{minimize}	& \ \ \ \ \ \	& y 													\notag \\
\mbox{subject to} 	& 		& a^s_ix + b^s_i - \Phi^s(m^s) \le y, \ \ \ \ \ \ \forall s \in {\cal S}^*	\notag \\
				&		& v_i \le x \le v_{i+1}.										\notag
\end{eqnarray}
This problem can be solved in $O(n^2)$ time by \cite{d84}.
Repeating the same operations over $i \in \{1, \ldots, n-1\}$, Phase 2 is completed in $O(n^3)$ time.

\begin{thm}
The minimax regret 1-median problem in dynamic path networks with uniform capacity can be solved in $O(n^3)$ time.
\label{thm:1}
\end{thm}

\section{Conclusion}
\label{sec:co}
In this paper, we address the minimax regret 1-median problem in dynamic path networks with uniform capacity and present an $O(n^3)$ time algorithm. 
Additionally, this is the first work that treats the minimax regret facility location problem in dynamic networks where the total time criterion is adopted.
Two natural questions immediately follow.
The first one is whether we can reduce the number of scenarios to be considered.
The other one is whether we can extend the problem to the $k$-median version with $k \ge 2$,
or the problem in more general networks.  


\clearpage

\section*{Appendices}

\subsection*{Appendix A: Proof of Lemma \ref{lem:wcsxy}}
\label{app:a}

We only treat the case of $\langle x, y \rangle \in P \times V$ satisfying $y < x$ since the other case is symmetric.

We first prove the following claim.
\begin{clm}
Given a pair $\langle x, y \rangle \in P \times V$ satisfying $y < x$,
there exists a worst case scenario for $\langle x, y \rangle$ such that 
the weight of every vertex $v_i \in [v_1, y]$ is $w^+_i$.
\label{clm:2}
\end{clm}

\noindent
{\it Proof of Claim \ref{clm:2}.} \
We here assume $y > v_1$: if $y = v_1$, the proof is straightforward.
Let $s_1$ be a worst case scenario for $\langle x, y \rangle$.
If there are more than one worst case scenario, 
we choose the one such that all weights are lexicographically maximized in the order of ascending indices
among all worst case scenarios.
We prove by contradiction: suppose $w^{s_1}_i < w^+_i$ for some vertex $v_i \in [v_1, y]$. 
Under a scenario $s_1$, let us perform the left-clustering for $x$ and $y$, respectively (refer to Section \ref{subsec:to}).
Performing the left-clustering for $x$,
let $c_x$ be a left cluster for $x$ including $v_i$, 
$\rho_x$ be the index of a vertex that corresponds to the head of $c_x$,
and $\sigma_x$ be the weight of $c_x$.
Also, performing the left-clustering for $y$, $c_y$, $\rho_y$ and $\sigma_y$ are similarly defined. 
When $c_y$ evacuates to $x$, 
by the definition of a cluster, the first unit of $c_y$ does not encounter any congestion in the interval $(v_{\rho_y}, y)$,
however, it may encounter congestions at some vertices in $[y, x)$.
If it does not encounter any congestion in $(v_{\rho_y}, x)$, 
$c_y$ is never merged to any other cluster for $x$,
which implies $\rho_x =\rho_y$ and $\sigma_x = \sigma_y$.
Otherwise, $c_y$ is eventually merged into a left cluster for $x$ headed by a vertex $v_{\rho_x} \in [y, x)$, i.e., $c_x$.
We thus consider two cases: 
[Case 1] $\rho_x = \rho_y$ and $\sigma_x = \sigma_y$; 
[Case 2] $\rho_y < y \le \rho_x$ and $\sigma_x > \sigma_y$.

Now let $s_2$ be a scenario obtained from $s_1$ by increasing the weight of $v_i$ by infinitesimally small $\delta > 0$, 
i.e., $w^{s_2}_i=w^{s_1}_i+\delta$ and $w^{s_2}_j=w^{s_1}_j$ for $j \not= i$.
We then show the following claim for the proof of Claim \ref{clm:2}.
\begin{clm}
While $s_1$ changes to $s_2$, 
the formation of left clusters for $x$ (resp. $y$) remains the same.
\label{clm:1}
\end{clm}
{\it Proof of Claim \ref{clm:1}.}  \
Performing the left-clustering for $x$ under $s_1$,
let $c'_x$ be a left cluster for $x$ immediately following $c_x$,
and $\rho'_x$
be the index of a vertex that corresponds to the head of $c'_x$.
Referring to (\ref{eq:ms2.6}), the following inequality holds:
\begin{eqnarray}
\tau(v_{\rho_x} - v_{\rho'_x}) > \frac{\sigma_x}{c} = \frac{\sum_{l=\rho'_x+1}^{\rho_x} w^{s_1}_l}{c}.
\label{eq:clm1.1}
\end{eqnarray}
Then, for a sufficiently small $\delta>0$, we have
\begin{eqnarray}
\tau(v_{\rho_x} - v_{\rho'_x}) > \frac{(\sum_{l=\rho'_x+1}^{\rho_x} w^{s_1}_l)+\delta}{c} = \frac{\sum_{l=\rho'_x+1}^{\rho_x} w^{s_2}_l}{c}.
\label{eq:clm1.2}
\end{eqnarray}
The inequality of (\ref{eq:clm1.1}) means that after performing the left-clustering for $x$ under $s_1$, the first unit of $c'_x$ does not catch up with the last unit of $c_x$ at $v_{\rho_x}$,
and by (\ref{eq:clm1.2}), this remark also holds even for $s_2$. 
Under $s_2$, if $c'_x$ is not merged to $c_x$, any other merge never occurs.
Thus, the formation of left clusters for $x$ does not change,
and similarly, it does not change for $y$. 
\qed

\medskip

By Claim \ref{clm:1} and the definitions of (\ref{eq:ms1}) and (\ref{eq:ms3}), we have 
\begin{eqnarray}
\Phi^{s_2}(x) - \Phi^{s_1}(x)	&=& \Phi^{s_2}_L(x) - \Phi^{s_1}_L(x) \notag \\
						&=& (\sigma_x+\delta)\tau(x - v_{\rho_x}) + \frac{(\sigma_x+\delta)^2}{2c} - \left\{\sigma_x\tau(x - v_{\rho_x}) + \frac{\sigma_x^2}{2c}\right\} \notag \\
						&=& \delta\tau(x - v_{\rho_x}) + \frac{\delta\sigma_x}{c} + \frac{\delta^2}{2c},
\label{eq:lem1.1}
\end{eqnarray}
and similarly, 
\begin{eqnarray}
\Phi^{s_2}(y) - \Phi^{s_1}(y)	&=& \delta\tau(y - v_{\rho_y}) + \frac{\delta\sigma_y}{c} + \frac{\delta^2}{2c}.
\label{eq:lem1.2}
\end{eqnarray}
Also by the definition of (\ref{eq:gap}), we have
\begin{eqnarray}
\Gamma^{s_2}(x, y) - \Gamma^{s_1}(x, y)	&=& \Phi^{s_2}(x) - \Phi^{s_2}(y) - (\Phi^{s_1}(x) - \Phi^{s_1}(y)) \notag \\
								&=& \Phi^{s_2}(x) - \Phi^{s_1}(x) - (\Phi^{s_2}(y) - \Phi^{s_1}(y)).
\label{eq:lem1.3}
\end{eqnarray}
From (\ref{eq:lem1.1}), (\ref{eq:lem1.2}) and (\ref{eq:lem1.3}), we can derive 
\begin{eqnarray}
\Gamma^{s_2}(x, y) - \Gamma^{s_1}(x, y)	&=& \delta\tau(x - v_{\rho_x}) + \frac{\delta\sigma_x}{c} + \frac{\delta^2}{2c} - \left\{\delta\tau(y - v_{\rho_y}) + \frac{\delta\sigma_y}{c} + \frac{\delta^2}{2c}\right\} \notag \\
								&=& \delta\tau(x - y - v_{\rho_x} + v_{\rho_y}) + \frac{\delta(\sigma_x-\sigma_y)}{c}.
\label{eq:lem1.4}
\end{eqnarray}

If Case 1 occurs, 
we can immediately see that the right side of (\ref{eq:lem1.4}) is greater than zero, i.e., $\Gamma^{s_2}(x, y) > \Gamma^{s_1}(x, y)$,
which contradicts that $s_1$ is a worst case scenario for $\langle x, y \rangle$.

If Case 2 occurs, 
performing the left-clustering for $x$ merges $c_y$ into $c_x$ as mentioned above,
and then, all units on $(v_{\rho_y}, v_{\rho_x}]$ are also merged into $c_x$.
Therefore, if we consider an input such that all units on $(v_{\rho_y}, v_{\rho_x})$ are moved to $v_{\rho_x}$,
the first unit of $c_y$ must catch up with the last unit of supply at $v_{\rho_x}$,
i.e.,
\begin{eqnarray}
\tau(v_{\rho_x}-v_{\rho_y}) &\le& \frac{\sum_{l=\rho_y+1}^{\rho_x} w^{s_1}_l}{c}.
\label{eq:lem1.5}
\end{eqnarray}
Since $c_x$ includes $c_y$ and all units on $(v_{\rho_y}, v_{\rho_x}]$, we have
\begin{eqnarray}
\sum_{l=\rho_y+1}^{\rho_x} w^{s_1}_l \le \sigma_x-\sigma_y.
\label{eq:lem1.6}
\end{eqnarray}
Note that in (\ref{eq:lem1.6}), the left side is less than the right side when $c_x$ also includes some clusters for $y$ following $c_y$.
From (\ref{eq:lem1.4}), (\ref{eq:lem1.5}) and (\ref{eq:lem1.6}), we can derive 
\begin{eqnarray}
\Gamma^{s_2}(x, y) - \Gamma^{s_1}(x, y)	\ge \delta\tau(x - y) > 0,
\label{eq:lem1.7}
\end{eqnarray}
which contradicts that $s_1$ is a worst case scenario for $\langle x, y \rangle$.
\qed

\medskip

If we consider a worst case scenario for $\langle x, y \rangle$ such that the weight of every vertex $v_i \in [v_1, y]$ is $w^+_i$ and weights of all other vertices in $(y, v_n]$ are lexicographically minimized in the order of descending indices,
the following claim is also proved in a similar manner as in the proof of Claim \ref{clm:2}.
\begin{clm}
Given a pair $\langle x, y \rangle \in P \times V$ satisfying $y < x$,
there exists a worst case scenario for $\langle x, y \rangle$ such that 
the weight of every vertex $v_i \in [v_1, y]$ is $w^+_i$ and 
the weight of every vertex $v_i \in [x, v_n]$ is $w^-_i$.
\label{clm:3}
\end{clm}

Now, let $s_3$ be a worst case scenario for $\langle x, y \rangle$ such that the weight of every vertex $v_i \in [v_1, y]$ is $w^+_i$, 
the weight of every vertex $v_i \in [x, v_n]$ is $w^-_i$, and
weights of all other vertices in the open interval $(y, x)$ are lexicographically maximized in the order of ascending indices.
Then, $s_3 \in {\cal S}_L$ can be proved.
We prove by contradiction: there exist two vertices $v_i, v_j \in (y, x)$ satisfying $i < j$ such that $w^{s_3}_i < w^+_i$ and $w^{s_3}_j > w^-_j$.
Let $s_4$ be a scenario obtained from $s_3$ by increasing the weight of $v_i$ by infinitesimally small $\delta > 0$ and
decreasing the weight of $v_j$ by the same $\delta$, 
i.e., $w^{s_4}_i=w^{s_3}_i+\delta$, $w^{s_4}_j=w^{s_3}_j-\delta$ and $w^{s_4}_k=w^{s_3}_k$ for $k \not= i, j$.
Then, we immediately see $\Phi^{s_4}(x)  \ge \Phi^{s_3}(x)$ and $\Phi^{s_4}(y) \le \Phi^{s_3}(y)$, therefore
\begin{eqnarray}
\Phi^{s_4}(x)-\Phi^{s_4}(y)	&\ge&	\Phi^{s_3}(x)-\Phi^{s_3}(y), \notag \ \ \ \ \mbox{i.e.,} \\
\Gamma^{s_4}(x, y)		&\ge&	\Gamma^{s_3}(x, y).
\label{eq:lem1.8}
\end{eqnarray}
The inequality of (\ref{eq:lem1.8}) implies that $s_4$ is also a worst case scenario for $\langle x, y \rangle$,
which contradicts the lexicographical maximality of weights on the open interval $(y, x)$ under $s_3$.
\qed

\subsection*{Appendix B: Proof of Lemma \ref{lem:wcsxy2}}
\label{app:b}
For a fixed pair $\langle x, y \rangle \in P \times V$ satisfying $y < x$, 
let us consider a worst case scenario 
in ${\cal S}_L$ such that the intermediate vertex is $v_i$ ($y \le v_i < x$).
We now consider the weight of $v_i$ as a variable $w \in [w^-_i, w^+_i]$,
and let 
$s(w)$ be a scenario in ${\cal S}_L$
such that the intermediate vertex is $v_i$ whose weight is $w$.
Then, let $\Gamma(w)$ denote the gap between $x$ and $y$ under $s(w)$, i.e.,
\begin{eqnarray}
\Gamma(w) = \Phi^{s(w)}(x)-\Phi^{s(w)}(y).
\label{eq:lem1.9}
\end{eqnarray}
Suppose that $s(w)$ is critically left-pseudo-bipartite for $y$ when $w = \omega_1, \ldots, \omega_p$,
where $p$ is a positive integer and $w^-_i = \omega_1 < \ldots < \omega_p = w^+_i$.
In the following, we prove that a function $\Gamma(w)$ is convex and piecewise-linear for $w \in [\omega_j, \omega_{j+1}]$ for every $j \in \{1, \ldots, p-1\}$.

Under a scenario $s(w)$, let us perform the left-clustering for $x$ and the right-clustering for $y$, respectively.
Performing the left-clustering for $x$,
let $c_x(w)$ be a left cluster for $x$ including $v_i$, 
$\rho_x(w)$ be the index of a vertex that corresponds to the head of $c_x(w)$,
and $\sigma_x(w)$ be the weight of $c_x$.
Also, performing the right-clustering for $y$, $c_y(w)$, $\rho_y(w)$ and $\sigma_y(w)$ are similarly defined. 

We first show the following claim.
\begin{clm}
$\Gamma(w)$ is continuous for $w \in [w^-_i, w^+_i]$.
\label{clm:4}
\end{clm}
\noindent
{\it Proof of Claim \ref{clm:4}.} \
The statement is equivalent to
\begin{eqnarray}
\lim_{\delta \to +0} \Gamma(w+\delta)	&=&	\Gamma(w) \quad \forall w \in [w^-_i, w^+_i), \quad \mbox{and} \label{eq:lem1.10} \\
\lim_{\delta \to +0} \Gamma(w-\delta)		&=&	\Gamma(w) \quad \forall w \in (w^-_i, w^+_i]. \label{eq:lem1.11}
\end{eqnarray}
We here prove (\ref{eq:lem1.10}) (the case of (\ref{eq:lem1.11}) is similarly treated).
By (\ref{eq:lem1.9}), we only have to prove
\begin{eqnarray}
\lim_{\delta \to +0} \Phi^{s(w+\delta)}(x) = \Phi^{s(w)}(x) \quad \forall w \in [w^-_i, w^+_i).
\label{eq:lem1.12}
\end{eqnarray}
Note that, similarly as in Claim \ref{clm:1}, while $s(w)$ changes to $s(w+\delta)$, 
the formation of left clusters for $x$ remains the same.
Therefore, by the definitions of (\ref{eq:ms1}) and (\ref{eq:ms3}), we have 
\begin{eqnarray}
\Phi^{s(w+\delta)}(x)	&=& \Phi^{s(w)}(x) +  \delta\tau(x - v_{\rho_x(w)}) + \frac{\delta\sigma_x(w)}{c} + \frac{\delta^2}{2c},
\label{eq:lem1.13}
\end{eqnarray}
which leads (\ref{eq:lem1.12}) by letting $\delta$ go to positive zero.
\qed

\medskip

For an integer $j \in \{1, \ldots, p-1\}$, we consider the right-derivative of $\Gamma(w)$ for $w \in [\omega_j, \omega_{j+1})$, i.e.,
\begin{eqnarray}
\Gamma'_+(w) = \lim_{\delta \to +0} \frac{\Gamma(w+\delta)-\Gamma(w)}{\delta}.
\label{eq:lem1.14}
\end{eqnarray}
Similarly to (\ref{eq:lem1.13}), we have
\begin{eqnarray}
\Phi^{s(w+\delta)}(y)	&=& \Phi^{s(w)}(y) +  \delta\tau(v_{\rho_y(w)}-y) + \frac{\delta\sigma_y(w)}{c} + \frac{\delta^2}{2c}. \label{eq:lem1.15} 
\end{eqnarray}
From (\ref{eq:lem1.9}), (\ref{eq:lem1.13}) and (\ref{eq:lem1.15}), we derive
\begin{eqnarray}
\Gamma(w+\delta)-\Gamma(w)	&=&	\Phi^{s(w+\delta)}(x)-\Phi^{s(w+\delta)}(y)- \left\{ \Phi^{s(w)}(x)-\Phi^{s(w)}(y) \right\} \notag \\
						&=&	\delta\tau(x + y - v_{\rho_x(w)} - v_{\rho_y(w)}) + \frac{\delta\left\{ \sigma_x(w)-\sigma_y(w) \right\}}{c},		
\label{eq:lem1.16} 
\end{eqnarray}
and by (\ref{eq:lem1.14}) and (\ref{eq:lem1.16}),
\begin{eqnarray}
\Gamma'_+(w)	=	\tau(x + y - v_{\rho_x(w)} - v_{\rho_y(w)}) + \frac{\sigma_x(w)-\sigma_y(w)}{c}.	
\label{eq:lem1.17} 
\end{eqnarray}
We here notice that as $w$ increases, $\rho_x(w)$ and $\rho_y(w)$ never change 
(even if left clusters for $x$ following $c_x(w)$ and right clusters for $y$ following $c_y(w)$ are merged to $c_x(w)$ and $c_y(w)$, respectively).
Also, since both of $c_x(w)$ and $c_y(w)$ include $v_i$, 
$\sigma_x(w)$ and $\sigma_y(w)$ can be represented as follows:
\begin{eqnarray}
\sigma_x(w)	&=& w + \sigma^j_x(w), \quad \mbox{and} \label{eq:lem1.18} \\
\sigma_y(w)	&=& w + \sigma^j_y(w), \label{eq:lem1.19} 
\end{eqnarray}
where $\sigma^j_x(w)$ and $\sigma^j_y(w)$ are functions of $w$.
In addition, $\sigma^j_x(w)$ increases only if a left cluster for $x$ following $c_x(w)$ is merged to $c_x(w)$, 
i.e., $\sigma^j_x(w)$ is an increasing step function,
and $\sigma^j_y(w)$ is a constant function of $w \in [\omega_j, \omega_{j+1})$ since the formation of right clusters for $y$ does not change over $w \in [\omega_j, \omega_{j+1})$
(recall the definition of critical left-pseudo-bipartite scenarios for $y$ in Section \ref{subsec:cpbs}).
From the above observations, and (\ref{eq:lem1.17}), (\ref{eq:lem1.18}) and (\ref{eq:lem1.19}), we derive that for $w \in [\omega_j, \omega_{j+1})$,
\begin{eqnarray}
\Gamma'_+(w)	=	\frac{\sigma^j_x(w)}{c} + const.,
\label{eq:lem1.20} 
\end{eqnarray}
which is an increasing step function.
By this fact and the continuity of $\Gamma(w)$ by Claim \ref{clm:4}, we have the following claim.
\begin{clm}
For an integer $j \in \{1, \ldots, p-1\}$, $\Gamma(w)$ is convex and piecewise-linear for $w \in [\omega_j, \omega_{j+1}]$.
\label{clm:5}
\end{clm}

By Claim \ref{clm:5}, a solution that maximizes $\Gamma(w)$ must be in $\{ \omega_1, \ldots, \omega_p \}$,
i.e., a worst case scenario for $\langle x, y \rangle$ is critically left-pseudo-bipartite for $y$.
\qed




\begin{thebibliography}{99}

	\bibitem{aags14}
		G. P. Arumugam, J. Augustine, M. J. Golin, P. Srikanthan,
		\newblock ``A polynomial time algorithm for minimax-regret evacuation on a dynamic path",
		\newblock {\em CoRR abs/1404.5448},
		\newblock arXiv:1404.5448.

	\bibitem{ab00} 
		I.~Averbakh and O.~Berman, 
		\newblock ``Algorithms for the robust $1$-center problem on a tree'', 
		\newblock {\em European Journal of Operational Research},
		\newblock 123(2), pp.~292-302, 2000.
	
	\bibitem{bk12}
		B.~Bhattacharya and T.~Kameda,
		\newblock ``A linear time algorithm for computing minmax regret $1$-median on a tree",
		\newblock {\em Proc. the 18th Annual International Computing and Combinatorics Conference} (COCOON 2012),
		\newblock LNCS 7434, pp.~1-12, 2012. 
		
	\bibitem{bk14}
		B. Bhattacharya, T. Kameda,
		\newblock ``Improved algorithms for computing minmax regret 1-sink and 2-sink on path network",
		\newblock {\em Proc. The 8th Combinatorial Optimization and Applications} (COCOA 2014),
		\newblock LNCS 8881, pp.~146-160, 2014.
	
	\bibitem{bks13}
		B.~Bhattacharya, T.~Kameda and Z.~Song,
		\newblock ``A linear time algorithm for computing minmax regret $1$-median on a tree network",
		\newblock {\em Algorithmica},
		\newblock pp.~1-20, 2013. 

	\bibitem{bgk08}
		G.~S.~Brodal, L.~Georgiadis and I.~Katriel,
		\newblock ``An $O(n \log n)$ version of the Averbakh-Berman algorithm for the robust median of a tree",
		\newblock {\em Operations Research Letters},
		\newblock 36(1), pp.~14-18, 2008. 

	\bibitem{cl98}
		B.~Chen and C.~Lin, 
		\newblock ``Minmax-regret robust 1-median location on a tree'', 
		\newblock {\em Networks},
		\newblock 31(2), pp.~93-103, 1998.

		
	\bibitem{chknsx13}
		S.~W.~Cheng, Y.~Higashikawa, N.~Katoh, G.~Ni, B.~Su and Y.~Xu,
		\newblock ``Minimax regret 1-sink location problems in dynamic path networks",
		\newblock {\em Proc. The 10th Annual Conference on Theory and Applications of Models of Computation} (TAMC 2013),
		\newblock LNCS 7876, pp.~121-132, 2013.


	\bibitem{c08} 
		E.~Conde, 
		\newblock ``A note on the minmax regret centdian location on trees'',
		\newblock {\em Operations Research Letters},
		\newblock 36(2), pp.~271-275, 2008.

	\bibitem{d84} 
		M.~E.~Dyer, 
		\newblock ``Linear time algorithms for two- and three-variable linear programs'',
		\newblock {\em SIAM Journal on Computing},
		\newblock 13(1), pp.~31-45, 1984.
		
	\bibitem{ff58}
		L.~R.~Ford~Jr., D.~R.~Fulkerson,
		\newblock ``Constructing maximal dynamic flows from static flows",
		\newblock {\em Operations Research},
		\newblock 6, pp.~419-433, 1958.
	
	\bibitem{h14}
		Y.~Higashikawa,
		\newblock ``Studies on the Space Exploration and the Sink Location under Incomplete Information towards Applications to Evacuation Planning",
		\newblock {\em Doctoral Dissertation},
		\newblock Kyoto University, 2014.
			
	\bibitem{hacknsx14}
		Y.~Higashikawa, J.~Augustine, S.~W.~Cheng, N.~Katoh, G.~Ni, B.~Su and Y.~Xu,
		\newblock ``Minimax Regret 1-Sink Location Problem in Dynamic Path Networks",
		\newblock {\em Theoretical Computer Science},
		\newblock DOI: 10.1016/j.tcs.2014.02.010, 2014.
	
	\bibitem{hgk14}
		Y.~Higashikawa, M.~J.~Golin, N.~Katoh,
		\newblock ``Minimax Regret Sink Location Problem in Dynamic Tree Networks with Uniform Capacity",
		\newblock {\em Proc. The 8th International Workshop on Algorithms and Computation} (WALCOM 2014),
		\newblock LNCS 8344, pp.~125-137, 2014.

	\bibitem{hgk14_2}
		Y.~Higashikawa, M.~J.~Golin, N.~Katoh,
		\newblock ``Multiple sink location problems in dynamic path networks",
		\newblock {\em Proc. The 10th International Conference on Algorithmic Aspects of Information and Management} (AAIM 2014),
		\newblock LNCS 8546, pp.~149-161, 2014.
		
	\bibitem{hgk14_3}
		Y.~Higashikawa, M.~J.~Golin, N.~Katoh,
		\newblock ``Minimax Regret Sink Location Problem in Dynamic Tree Networks with Uniform Capacity",
		\newblock {\em Journal of Graph Algorithms and Applications},
		\newblock 18(4), pp.~539-555, 2014.

	\bibitem{hgk14_4}
		Y.~Higashikawa, M.~J.~Golin, N.~Katoh,
		\newblock ``Multiple Sink Location Problems in Dynamic Path Networks",
		\newblock {\em Theoretical Computer Science},
		\newblock DOI:10.1016/j.tcs.2015.05.053.
	
			


	\bibitem{mumf06}
		S.~Mamada, T.~Uno, K.~Makino and S.~Fujishige,
		\newblock ``An $O(n \log^2 n)$ algorithm for the optimal sink location problem in dynamic tree networks",
		\newblock {\em Discrete Applied Mathematics},
		\newblock 154(16), pp.~2387-2401, 2006. 



	\bibitem{w13}
		H.~Wang, 
		\newblock ``Minmax regret 1-facility location on uncertain path networks",
		\newblock {\em Proc. The 24th International Symposium on Algorithms and Computation (ISAAC 2013)},
		\newblock LNCS 8283, pp.~733-743, 2013.

\end{thebibliography}
\end{document}